\documentclass[11pt]{article}
\usepackage{amsfonts,amsthm, amsmath}
\usepackage{epsfig}
\usepackage{makeidx}
\usepackage{commath}
\usepackage{subfigure}
\usepackage{graphicx}

\makeindex
\newcommand{\ncom}{\newcommand}
\ncom{\ul}{\underline}
\ncom{\beq}{\begin{equation}}
\ncom{\eeq}{\end{equation}}
\ncom{\bea}{\begin{eqnarray*}}
\ncom{\eea}{\end{eqnarray*}}
\ncom{\beqa}{\begin{eqnarray}}
\ncom{\eeqa}{\end{eqnarray}}
\ncom{\nno}{\nonumber}
\ncom{\non}{\nonumber}
\ncom{\ds}{\displaystyle}
\ncom{\half}{\frac{1}{2}}
\ncom{\mbx}{\makebox{.25cm}}
\ncom{\hs}{\mbox{\hspace{.25cm}}}
\ncom{\rar}{\rightarrow}
\ncom{\Rar}{\Rightarrow}
\ncom{\noin}{\noindent}
\ncom{\bc}{\begin{center}}
\ncom{\ec}{\end{center}}
\ncom{\sz}{\scriptsize}
\ncom{\rf}{\ref}
\ncom{\s}{\sqrt{2}}
\ncom{\sgm}{\sigma}
\ncom{\Sgm}{\Sigma}
\ncom{\psgm}{\sigma^{\prime}}
\ncom{\dt}{\delta}
\ncom{\Dt}{\Delta}
\ncom{\lmd}{\lambda}
\ncom{\Lmd}{\Lambda}
\ncom{\Th}{\Theta}
\ncom{\e}{\eta}
\ncom{\eps}{\epsilon}
\ncom{\pcc}{\stackrel{P}{>}}
\ncom{\lp}{\stackrel{L_{p}}{>}}
\ncom{\dist}{{\rm\,dist}}
\ncom{\sspan}{{\rm\,span}}
\ncom{\re}{{\rm Re\,}}
\ncom{\im}{{\rm Im\,}}
\ncom{\sgn}{{\rm sgn\,}}
\ncom{\ba}{\begin{array}}
\ncom{\ea}{\end{array}}
\ncom{\hone}{\mbox{\hspace{1em}}}
\ncom{\htwo}{\mbox{\hspace{2em}}}
\ncom{\hthree}{\mbox{\hspace{3em}}}
\ncom{\hfour}{\mbox{\hspace{4em}}}
\ncom{\vone}{\vskip 2ex}
\ncom{\vtwo}{\vskip 4ex}
\ncom{\vonee}{\vskip 1.5ex}
\ncom{\vthree}{\vskip 6ex}
\ncom{\vfour}{\vspace*{8ex}}
\ncom{\integ}[4]{\int_{#1}^{#2}\,{#3}\,d{#4}}
\ncom{\vspan}[1]{{{\rm\,span}\{ #1 \}}}
\ncom{\dm}[1]{ {\displaystyle{#1} } }
\ncom{\ri}[1]{{#1} \index{#1}}

\newtheorem{proposition}{Proposition}[section]

\newtheoremstyle
    {remarkstyle}
    {}
    {11pt}
    {}
    {}
    {\bfseries}
    {:}
    {     }
    {\thmname{#1} \thmnumber{#2} }

\theoremstyle{remarkstyle}

\begin{document}

\newpage

\begin{center}
{\Large \bf Normal Inverse Gaussian Autoregressive Model Using EM Algorithm}\\
\end{center}
\vone
\begin{center}
{\bf  Monika S. Dhull$^a$ and Arun Kumar$^{a*}$}\\
$^{a}${\it Department of Mathematics,
Indian Institute of Technology Ropar,\\ Rupnagar, Punjab 140001, India.}\\
\end{center}
\makeatletter{\renewcommand*{\@makefnmark}{}
\footnotetext{*Email: arun.kumar@iitrpr.ac.in}\makeatother}
\vtwo
\begin{center}
\noindent{\bf Abstract}
\end{center}
In this article, normal inverse Gaussian (NIG) autoregressive model is introduced. The parameters of the model are estimated using Expectation Maximization (EM) algorithm. The efficacy of the EM algorithm is shown using simulated and real world financial data. It is shown that NIG autoregressive model fit very well on the considered financial data and hence could be useful in modeling of various real life time-series data.

\vone \noindent{\it Key words:} Normal inverse Gaussian distribution, autoregressive model, EM algorithm, Monte Carlo simulations.
\vone
\vone
\setcounter{equation}{0}

\section{Introduction}
The most simple and intuitive time series model encountered in the time-series modeling is the standard first-order autoregressive process which is also denoted by AR(1) \cite{Tsay2005}. In the AR(1) model each new entry is the sum of two terms one is proportional to the previous entry and  the another is a white-noise process also called the error term or the innovation term. In a white-noise process variables are independent and identically distributed (iid) with zero mean. If the variables are Gaussian this is called Gaussian white noise. 
However in financial markets the observed time series (log-returns) distributions are non-Gaussian and have tails heavier than Gaussian and lighter than power law. These kind of distributions are also called semi heavy-tailed distributions ( see e.g. \cite{Rachev2003, Cont2004, Omeya2018}). For a literature survey on different innovation distributions or marginal distributions in the AR(1) model \cite{Grunwald1996}. 
The AR(1) models with non-Gaussian innovation terms are very well considered in the literature. Sim \cite{Sim1990} considered AR(1) model with Gamma process as the innovation term. For AR models with innovations following a Student's t-distribution see e.g. \cite{Tiku2000, Tarami2003, Christmas2011, Nduka2018} and references therein. Note that $t$-distribution is used in modeling of asset returns \cite{Heyde2005}.

Normal inverse Gaussian (NIG) distribution introduced by Barndorff-Nielsen \cite{Barndorff1997a} is a semi heavy tailed distribution with tails heavier than the Gaussian but lighter than the power law tails. NIG distribution is defined as the normal variance-mean mixture where the mixing distribution is the inverse Gaussian distribution. NIG distributions and processes are used to model the returns from the financial time-series see e.g. \cite{Kalemanova2007, Barndorff1997b}. A more general distribution which is obtained by taking normal variance-mean mixture with mixing distribution as generalised inverse Gaussian distribution is called generalised hyperbolic distribution \cite{Barndorff2013}.

In this article, we introduce an AR(1) model with NIG innovation terms. Due to heavy-tailedness of the innovation term it can model large jumps in the observed data. The introduced process very well model the Google stock price time-series. The non-Gaussian behaviour of the innovation term of the Google stock price is shown using QQ plot and Kolmogorov-Smirnov test. From a market risk management perspective obtaining
reasonable extreme observation levels 
is a crucial objective in modeling. Since, it capture the market extreme movements which a Gaussian based model doesn't capture. The parameters of the model are estimated using EM algorithm. The efficacy of the estimation procedure is shown on the simulated data. The rest of the paper is organised as follows. In Section 2, the NIG autoregressive model is defined along with important properties of the NIG distribution. Section 3, discuss the estimation procedure of the parameters of the introduced model using EM algorithm. The efficacy of the estimation procedure on simulated data and the real world financial data application is discussed in Section 4.  Last Section concludes.

\setcounter{equation}{0}
\section{NIG Autoregressive Model}
We introduce the AR(1) model with iid normal inverse Gaussian (NIG) innovations  $\varepsilon_{t}$. Consider a univariate time series $ Y_1, Y_2, \cdots, Y_t$ following 

\begin{equation}\label{main_model}
Y_t = \rho Y_{t-1} + \varepsilon_{t},
\end{equation} 
where $\varepsilon_{t}\sim$ NIG$(\alpha, \beta, \mu, \delta)$. NIG is a semi heavy tailed distribution which can be obtained as normal mean-variance mixture with inverse Gaussian as mixing distribution.
The conditional distribution of $Y_t$ given all the parameters $\rho, \alpha, \beta, \mu, \text{and } \delta$ conditioning on all the preceding data $\mathcal{F}_{t-1}$ (i.e. $Y_1, Y_2, \cdots, Y_{t-1}$) only depend on previous sample $Y_{t-1}$ and has the form
\begin{align*}
p(Y_{t}|\rho, \alpha, \beta, \mu, \delta, \mathcal{F}_{t-1}) = p(Y_{t}|\rho, \alpha, \beta, \mu, \delta, Y_{t-1})
 = f_{t}(y_t; \rho y_{t-1}, \alpha, \beta, \mu, \delta),
\end{align*}
where $f_{t}(\cdot)$ denotes the probability density function (pdf) of NIG$(\alpha, \beta, \mu, \delta),$ which is given by 
$$f(x; \alpha, \beta, \mu, \delta) = \frac{\alpha}{\pi}\exp\left(\delta \sqrt{\alpha^2 - \beta^2} - \beta \mu\right)\phi(x)^{-1/2}K_{1}(\delta \alpha \phi(x)^{1/2})\exp(\beta x),\;x\in \mathbb{R}, $$ 
where $\phi(x) = 1+[(x-\mu)/\delta ]^2$, $\alpha = \sqrt{\gamma^2+\beta^2}$ and $K_\nu(x)$ denotes the modified Bessel function of the third kind of order $\nu$ evaluated at $x$ and is defined by
$$
K_\nu(x) = \frac{1}{2}\int_{0}^{\infty}y^{\nu-1}e^{-\frac{1}{2}x(y+y^{-1})}dy.
$$
Further, the following properties of $K_\nu(\cdot)$ are used in subsequent calculations.
$$K_{-\nu}(x) = K_{\nu}(x);\;\;\; K_{2}(x) = K_{0}(x) + \frac{2}{x}K_{1}(x).$$ 
A pdf $h(x)$ is called a semi-heavy tailed pdf if
$$
h(x) \sim e^{-cx} g(x),\; c>0,
$$
where $g$ is a regularly varying function \cite{Omeya2018}. Recall that a positive function $g$ is regularly varying with index $\alpha$ if
$$
\lim_{x\rightarrow\infty}\frac{g(dx)}{g(x)} = d^{\alpha}, \; d>0.
$$
Using $K_{\nu}(\omega) \sim \sqrt{\frac{\pi}{2}}e^{-\omega}\omega^{-1/2}$, as $\omega \rightarrow \infty$ \cite{Jorgensen1982}, we have
\begin{align*}
f(x; \alpha, \beta, \mu, \delta) &\sim  \frac{\alpha}{\pi} e^{(\delta \sqrt{\alpha^2 - \beta^2} - \beta \mu)}\phi(x)^{-1/2} \sqrt{\frac{\pi}{2}} e^{-\delta\alpha\sqrt{\phi(x)}} \left(\delta\alpha \sqrt{\phi(x)}\right)^{-1/2}e^{\beta x}\\
& = \frac{\alpha}{\pi} e^{(\delta \sqrt{\alpha^2 - \beta^2} - \beta \mu)}\sqrt{\frac{\pi}{2}}(\delta\alpha)^{-1/2}\phi(x)^{-3/4}e^{-\delta\alpha\sqrt{\phi(x)}} e^{\beta x}\\
&\sim  \frac{\alpha}{\pi} e^{(\delta \sqrt{\alpha^2 - \beta^2} - \beta \mu)}\sqrt{\frac{\pi}{2}}(\delta\alpha)^{-1/2} \delta^{3/4}x^{-3/4}e^{\mu\alpha} e^{-(\alpha-\beta)x},\;\alpha>\beta,
\end{align*}
as $x\rightarrow \infty.$ Thus the pdf of NIG is semi-heavy tailed. Due to exponentially decaying tails the moments of all order exist for NIG distribution. The NIG distributed random variable can be generated from its normal variance-mean mixture form, such that
\begin{align}\label{mean_variance_mixture}
X = \mu + \beta G + \sqrt{G} Z,
\end{align} where $Z$ is standard normal i.e. $Z\sim N(0,1)$ and $G$ has an inverse Gaussian distribution (IG) with parameters $\gamma$ and $\delta$ denoted by $G\sim$IG$(\gamma, \delta)$. The mixing distribution IG$(\gamma, \delta)$ having pdf
\begin{equation}\label{IG_density}
g(x; \gamma, \delta) = \frac{\delta}{\sqrt{2\pi}}\exp(\delta \gamma)x^{-3/2}\exp\left(-\frac{1}{2}\left(\frac{\delta^2}{x} + \gamma^2x\right)\right),\;x>0.
\end{equation}
This form of NIG distribution makes it suitable to use EM algorithm for ML estimation of parameters. For AR(1) model, consider that the innovations $\varepsilon_{t}$ follows NIG distribution and has the form defined in \eqref{main_model}. The conditional distribution $\varepsilon_{t}| G_t = g \sim N(\mu + \beta g, g),$ here $G_t$ are independent copies of $G \sim $ IG$(\gamma, \delta)$.

\section{Parameter Estimation Using EM Algorithm}
In this section, we provide a step-by-step procedure to estimate the parameters of the proposed normal inverse Gaussian autoregressive model of order 1 called NIGAR(1) model using EM algorithm. The EM algorithm for NIG innovation is based on \cite{Karlis2002}. Note that, EM algorithm is a general iterative algorithm for model parameter estimation by maximizing the likelihood in presence of missing data. An alternative to numerical optimization of the likelihood function, the EM algorithm was introduced in 1977 by Dempster, Laird and Rubin \cite{Dempster1977}. The EM algorithm is popularly used in estimating Gaussian mixture models (GMMs), estimating hidden Markov models (HMMs) and model based data clustering. Recently, EM algorithm is also used in parameter estimation of time-series model \cite{Metaxoglou2007, Kim2008}.  Some extensions of EM algorithm include expectation conditional maximization (ECM) algorithm proposed by Meng and Rubin \cite{Meng1993}; expectation conditional maximization either (ECME) algorithm of Liu and Rubin \cite{Liu1995}. A detailed discussion on the theory of EM algorithm and its extensions is given in McLachlan and Krishnan \cite{McLachlan2007}. The EM algorithm iterates between two steps:
\begin{itemize}
\item [] {\it E-Step:} In this step, create a function $Q(\theta |\theta^{(k)})$ to compute the expectation of log-likelihood of unobserved/complete data $(X, G)$ with respect to the conditional distribution of $G$ given $X$, using the current estimate for the parameters.$$Q(\theta|\theta^{(k)}) = \mathbb{E}_{G|X,\theta^{(k)}}[\log f(X,G|\theta)|X, \theta^{(k)}].$$
\item[] {\it M-Step:} Compute the parameters by maximizing the expected log-likelihood of complete data found on the {\it E-step} such that
$$\theta^{(k+1)} = \operatorname*{argmax}_\theta Q(\theta|\theta^{(k)}).$$
\end{itemize}
The algorithm is proven to be numerically stable. Also, as a consequence of Jensen's inequality, log-likelihood function at the updated parameters $\theta^{(k+1)}$ will not be less than that at the current values $\theta^{(k)}$. Although there is always a concern that the algorithm might get stuck at local extrema, but it can be handled by starting from different initial values and comparing the solutions. In next proposition, we provide the estimates of the parameters if NIGAR(1) model using EM algorithms.

\begin{proposition}\label{NIG_AR(1)} 
Suppose the innovation terms $\varepsilon_t$ in the AR(1) time-series model $Y_t = \rho Y_{t-1} + \varepsilon_t$ are from NIG$(\alpha, \beta, \mu, \delta)$ distribution. Then, the ML estimates of the model parameters using EM algorithm are as follows:

{ \tiny
\begin{align}
\hat{\rho} &= \frac{\left(n\displaystyle\sum_{t=1}^{n}w_{t} y_{t}y_{t-1} - \sum_{t=1}^{n}w_{t} y_{t}\sum_{t=1}^{n}y_{t-1}\right)\left(\displaystyle\sum_{t=1}^{n}s_{t}\sum_{t=1}^{n}w_{t} - n^2\right) 
+ \left(n\displaystyle\sum_{t=1}^{n}w_{t} y_{t-1} - \sum_{t=1}^{n}w_{t} \sum_{t=1}^{n}y_{t-1}\right)\left(n\displaystyle\sum_{t=1}^{n}y_{t} - \sum_{t=1}^{n}w_{t} y_{t}\sum_{t=1}^{n}s_{t}\right)}{\left(n\displaystyle\sum_{t=1}^{n}w_{t} y_{t-1} - \sum_{t=1}^{n}w_{t} \sum_{t=1}^{n}y_{t-1}\right)\left(n\displaystyle\sum_{t=1}^{n}y_{t-1} - \sum_{t=1}^{n}w_t y_{t-1}\sum_{t=1}^{n}s_t\right) + \left(n\displaystyle\sum_{t=1}^{n}w_{t}y_{t-1}^{2} - \sum_{t=1}^{n}w_{t}y_{t-1}\sum_{t=1}^{n}y_{t-1}\right)\left(\displaystyle\sum_{t=1}^{n}s_{t}\sum_{t=1}^{n}w_{t} - n^2\right)};
\end{align}

\begin{align}
\hat {\mu} &= \frac{n\displaystyle\sum_{t=1}^{n}w_{t} y_{t}y_{t-1} - \sum_{t=1}^{n}w_{t} y_{t}\sum_{t=1}^{n}y_{t-1} - \rho \left(n\displaystyle\sum_{t=1}^{n}w_{t}y_{t-1}^{2} - \sum_{t=1}^{n}w_{t}y_{t-1}\sum_{t=1}^{n}y_{t-1}\right)}{\left(n\displaystyle\sum_{t=1}^{n}w_{t} y_{t-1} - \sum_{t=1}^{n}w_{t} \sum_{t=1}^{n}y_{t-1}\right)};\\
\hat {\beta} &= \frac{ \displaystyle\sum_{t=1}^{n}w_{t} y_{t} -\rho \sum_{t=1}^{n}w_t y_{t-1}  - \mu \sum_{t=1}^n w_{t} }{n};\\
\hat {\delta} &= \sqrt{\frac{\bar{s}}{(\bar{s} \bar{w} - 1)}},\;\; \hat {\gamma} = \frac{\delta}{\bar{s}},\;\mbox{and}\;
\hat{\alpha} = (\gamma^2 + \beta^2)^{1/2},
\end{align} }
where $\bar{s} = \frac{\sum_{t=1}^{n}s_t}{n}, \bar{w} = \frac{\sum_{t=1}^{n}w_t}{n}. $
\end{proposition}
\begin{proof}
For AR(1) model, to implement EM algorithm let $(\varepsilon_{t}, G_{t}),$ for $t = 1, 2, ..., n$ denote the complete data. The observable data $\varepsilon_{t}$ is assumed to be from NIG$(\alpha, \beta, \mu, \delta)$ and the unobserved data $G_{t}$ follows IG$(\gamma, \delta)$. Also, given the time series data $y_1, y_2, ..., y_n$, the innovations term can be rewritten as
$$\varepsilon_{t} = y_t - \rho y_{t-1}, \text{ for } t = 1, 2, ..., n.$$
Note that $\varepsilon| G = g \sim N(\mu + \beta g, g)$ and the conditional pdf is
\begin{align*}
f(\varepsilon=\varepsilon_t| G = g_t) 
= \frac{1}{\sqrt{2\pi g_t}}\exp\left(-\frac{1}{2g_t}(y_t - \rho y_{t-1} -\mu -\beta g_t)^2\right). 
\end{align*} 
Let $\tilde{\theta} = (\alpha, \beta, \mu, \delta, \rho)$ denote the vector of parameters we need to estimate. At {\it E-step}, we need to find the conditional expectation of log-likelihood of unobserved/complete data $(\varepsilon, G)$ with respect to the conditional distribution of $G$ given $\varepsilon$. 
In our case, unobserved data is from IG$(\gamma, \delta)$, or we can say prior distribution for G is IG$(\gamma, \delta)$, therefore, the posterior distribution is generalised inverse Gaussian (GIG) distribution. The distribution of $G | \varepsilon, \alpha, \beta, \mu, \delta, \rho$ is GIG$( -1, \delta \sqrt{\phi(\epsilon)}, \alpha)$. The conditional first moments and inverse first moment are as follows
\begin{align*}
\mathbb{E}(G|\epsilon) &= \frac{\delta \phi(\epsilon)^{1/2}}{\alpha}\frac{K_0(\alpha \delta \phi(\epsilon)^{1/2})}{K_1(\alpha \delta \phi(\epsilon)^{1/2})},\\
\mathbb{E}(G^{-1}|\epsilon) &= \frac{\alpha}{\delta \phi(\epsilon)^{1/2}}\frac{K_{-2}(\alpha \delta \phi(\epsilon)^{1/2})}{K_{-1}(\alpha \delta \phi(\epsilon)^{1/2})}.
\end{align*}

\noindent These expectations will be useful in finding the conditional expectation of the log-likelihood function.
The complete data likelihood is given by
\begin{align*}
L(\theta) &= \prod_{t=1}^{n}f(\epsilon_{t}, g_t) = \prod_{t=1}^{n}f_{\varepsilon|G}(\epsilon_{t} | g_t) f_{G}(g_t)\\
&= \prod_{t=1}^{n} \frac{\delta}{2 \pi g_{t}^2}\exp(\delta \gamma) \exp\left(-\frac{\delta^2}{2g_t} - \frac{\gamma^2 g_t}{2} - \frac{g_t^{-1}}{2}(y_{t} - \rho y_{t-1} - \mu)^2 - \frac{\beta^{2} g_t}{2} + \beta(y_t - \rho y_{t-1} - \mu)\right).
\end{align*}  

\noindent The log likelihood function will be
\begin{align*}
l(\theta) &= n\log(\delta) - n\log(2 \pi) + n\delta \gamma - n\beta \mu -2\sum_{t=1}^{n}\log(g_t) - \frac{\delta^2}{2} \sum_{t=1}^{n} g_t^{-1} \\
&- \frac{\gamma^2}{2} \sum_{t=1}^{n} g_t  - \frac{1}{2}\sum_{t=1}^{n} g_t^{-1}(y_t - \rho y_{t-1} - \mu)^2 - \frac{\beta^2}{2} \sum_{t=1}^{n} g_t + \beta \sum_{t=1}^{n}(y_t - \rho y_{t-1})
\end{align*}
Now to implement the {\it E-step} of EM algorithm, we need to compute $Q(\theta|\theta^k)$, which is the expected value of complete data log likelihood and can be expressed as 

\begin{align*}
Q(\theta|\theta^{(k)}) &= \mathbb{E}_{G|\varepsilon,\theta^{(k)}}[\log f(\varepsilon,G|\theta)|\varepsilon, \theta^{(k)}]= \mathbb{E}_{G|\varepsilon,\theta^{(k)}}[L(\theta|\theta^{(k)})]\\
&= n\log \delta + n\delta \gamma - n\beta \mu - n\log(2 \pi) - 2\sum_{t=1}^{n}\mathbb{E}(\log g_t| \epsilon_t, \theta^{(k)}) - \frac{\delta^2}{2}\sum_{t=1}^{n} w_t\\
&- \frac{\gamma^2}{2}\sum_{t=1}^{n} s_t - \frac{\beta^2}{2}\sum_{t=1}^{n} s_t + \beta \sum_{t=1}^n(y_t - \rho y_{t-1}) - \frac{1}{2}\sum_{t=1}^{n}(y_t - \rho y_{t-1} - \mu)^2 w_t,
\end{align*}
where, $s_t = \mathbb{E}_{G|\varepsilon,\theta^{(k)}}(g_t|\epsilon_t, \theta^{(k)})$ and 
$w_t = \mathbb{E}_{G|\varepsilon,\theta^{(k)}}(g_t^{-1}|\epsilon_t, \theta^{(k)})$. At {\it M-step}, to update the parameters, maximize the $Q$ function by solving the following equations:
\begin{align*}
    \frac{\partial{Q}}{\partial{\rho}} &=  \sum_{t=1}^{n}w_{t} y_{t}y_{t-1} -\beta \sum_{t=1}^{n}y_{t-1}  - \mu \sum_{t=1}^{n}w_{t}y_{t-1} - \rho \sum_{t=1}^{n}w_{t} y_{t-1}^{2},\\
    \frac{\partial{Q}}{\partial{\mu}} &= -n\beta + \sum_{t=1}^{n}w_{t} y_{t} - \mu \sum_{t=1}^{n}w_{t} - \rho \sum_{t=1}^{n}w_{t} y_{t-1},\\
    \frac{\partial{Q}}{\partial{\beta}} &= -n\mu + \sum_{t=1}^{n}y_{t} - \beta \sum_{t=1}^{n}s_{t} - \rho \sum_{t=1}^{n}y_{t-1}, \\
    \frac{\partial{Q}}{\partial{\delta}} &= n\gamma + \frac{n}{\delta} - \delta \sum_{t=1}^{n}w_{t}, \\
    \frac{\partial{Q}}{\partial{\gamma}} &= n\delta - \gamma\sum_{t=1}^{n}s_{t}.
\end{align*}

\noindent Solving the above equations, the following estimates of the parameters are obtained
{ \tiny
\begin{align*}
\hat{\rho} &= \frac{\left(n\displaystyle\sum_{t=1}^{n}w_{t} y_{t}y_{t-1} - \sum_{t=1}^{n}w_{t} y_{t}\sum_{t=1}^{n}y_{t-1}\right)\left(\displaystyle\sum_{t=1}^{n}s_{t}\sum_{t=1}^{n}w_{t} - n^2\right) 
+ \left(n\displaystyle\sum_{t=1}^{n}w_{t} y_{t-1} - \sum_{t=1}^{n}w_{t} \sum_{t=1}^{n}y_{t-1}\right)\left(n\displaystyle\sum_{t=1}^{n}y_{t} - \sum_{t=1}^{n}w_{t} y_{t}\sum_{t=1}^{n}s_{t}\right)}{\left(n\displaystyle\sum_{t=1}^{n}w_{t} y_{t-1} - \sum_{t=1}^{n}w_{t} \sum_{t=1}^{n}y_{t-1}\right)\left(n\displaystyle\sum_{t=1}^{n}y_{t-1} - \sum_{t=1}^{n}w_t y_{t-1}\sum_{t=1}^{n}s_t\right) + \left(n\displaystyle\sum_{t=1}^{n}w_{t}y_{t-1}^{2} - \sum_{t=1}^{n}w_{t}y_{t-1}\sum_{t=1}^{n}y_{t-1}\right)\left(\displaystyle\sum_{t=1}^{n}s_{t}\sum_{t=1}^{n}w_{t} - n^2\right)};
\end{align*}

\begin{align*}
\hat {\mu} &= \frac{n\displaystyle\sum_{t=1}^{n}w_{t} y_{t}y_{t-1} - \sum_{t=1}^{n}w_{t} y_{t}\sum_{t=1}^{n}y_{t-1} - \rho \left(n\displaystyle\sum_{t=1}^{n}w_{t}y_{t-1}^{2} - \sum_{t=1}^{n}w_{t}y_{t-1}\sum_{t=1}^{n}y_{t-1}\right)}{\left(n\displaystyle\sum_{t=1}^{n}w_{t} y_{t-1} - \sum_{t=1}^{n}w_{t} \sum_{t=1}^{n}y_{t-1}\right)};\\
\hat {\beta} &= \frac{ \displaystyle\sum_{t=1}^{n}w_{t} y_{t} -\rho \sum_{t=1}^{n}w_t y_{t-1}  - \mu \sum_{t=1}^n w_{t} }{n};\\
\hat {\delta} &= \sqrt{\frac{\bar{s}}{(\bar{s} \bar{w} - 1)}},\;\; \hat {\gamma} = \frac{\delta}{\bar{s}},\;\mbox{and}\;
\hat{\alpha} = (\gamma^2 + \beta^2)^{1/2},
\end{align*} }
where $\bar{s} = \frac{\sum_{t=1}^{n}s_t}{n}, \bar{w} = \frac{\sum_{t=1}^{n}w_t}{n}. $
\end{proof}

\begin{proposition}
For the time-series defined in \eqref{main_model}, we have
\begin{align*}
\mathbb{E}(Y_t) = \frac{\mu\gamma+\delta\beta}{\gamma}\left(\frac{1-\rho^{t+1}}{1-\rho}\right);\;\;\; {\rm Var}(Y_t) = \frac{\delta \alpha^2}{\gamma^3}\left(\frac{1-\rho^{2n+2}}{1-\rho^2}\right).
\end{align*}
\end{proposition}

\begin{proof}
For $\varepsilon_t \sim NIG(\alpha, \beta, \mu, \delta)$, we have $\mathbb{E}(\varepsilon_t) = \frac{\mu\gamma + \delta\beta}{\gamma}$ and $\mbox{Var}(\varepsilon_t) = \frac{\delta\alpha^2}{\gamma^3}.$
For an AR(1) time series, it follows 
\begin{align*}
\mathbb{E}(Y_t) & = \rho \mathbb{E}(Y_{t-1}) + \mathbb{E}\epsilon_t = \mathbb{E}\epsilon_t + \rho \mathbb{E}\epsilon_{t-1} + \rho^2 \mathbb{E}\epsilon_{t-2}+\cdots+ \rho^t\mathbb{E}\epsilon_1\\ & = \mathbb{E}(\epsilon_1)\left(\frac{1-\rho^{t+1}}{1-\rho}\right) = \frac{\mu\gamma+\delta\beta}{\gamma}\left(\frac{1-\rho^{t+1}}{1-\rho}\right).
\end{align*}
Using law of total variance to calculate ${\rm Var}(Y_t)$, yields to 
\begin{align*}
    {\rm Var}(Y_t)&= \mathbb{E}[{\rm Var}(Y_t|Y_{t-1}= y_{t-1})] + {\rm Var}[\mathbb{E}(Y_t|Y_{t-1}= y_{t-1})]\\
    &= \mathbb{E}[{\rm Var}(\rho Y_{t-1}+\varepsilon_t | Y_{t-1}= y_{t-1})] + {\rm Var}[\mathbb{E}(\rho Y_{t-1}+\varepsilon_t|Y_{t-1}=y_{t-1})]\\
    &= \mathbb{E}[{\rm Var}(\varepsilon_t | Y_{t-1}=y_{t-1})] + {\rm Var}[\mathbb{E}(\varepsilon_t) + \rho Y_{t-1}]\\
    &= {\rm Var}(\varepsilon_t) + {\rm Var}(\rho Y_{t-1})
    = {\rm Var}(\varepsilon_t) + \rho^2 {\rm Var}(Y_{t-1}).
    \end{align*}

\noindent Recursively using the above relation and since $\varepsilon_t$ are iid we can write,
\begin{align*}
    {\rm Var(Y_t)} &= {\rm Var}(\varepsilon_t) + (\rho^2+\rho^4+...+\rho^{2t}) {\rm Var}(\varepsilon_{t})\\ 
    &=  (1+\rho^2+\rho^4+...+\rho^{2n}) {\rm Var}(\varepsilon_{t})\\
    &= \frac{1-\rho^{2t-2}}{1-\rho^2} {\rm Var}(\varepsilon_{t})= \frac{\delta \alpha^2}{\gamma^3}\left(\frac{1-\rho^{2t-2}}{1-\rho^2}\right).
\end{align*}
\end{proof}



\noindent In \cite{Karlis2002}, an EM algorithm was developed to estimate the parameters of NIG distribution. 
The estimate of the parameter $\rho$ can be evaluated directly by applying the conditional least square method on the time series data ${y_{t}}$ which  is obtained by minimizing 
 \begin{equation}\label{CLS1}
 \sum_{i=1}^{n}\left[Y_{t+1}-\mathbb{E}(Y_{t+1}|Y_t,Y_{t-1},\cdots,Y_1)\right]^2,
 \end{equation}
 with respect to $\tilde{\theta}$ and leads to  
\begin{align}\label{cls_rho}
\hat{\rho} = \frac{\sum_{t=0}^{n}(y_{t} - \overline{y})(y_{t+1} - \overline{y})}{\sum_{t=0}^{n}(y_{t} - \overline{y})^2}.
\end{align}
    
\noindent For remaining parameters $\theta = (\alpha, \beta, \mu, \delta)$, we will use the EM algorithm proposed by Karlis \cite{Karlis2002} for ML estimation of NIG distribution. Using the steps given in \cite{Karlis2002}, following estimates of the parameters are obtained
\begin{align*}
\hat {\beta} &= \frac{\hat{\rho}\displaystyle \sum_{t=1}^{n}w_t y_{t-1} - \sum_{t=1}^{n}w_{t} y_{t} + \bar{y}\sum_{t=1}^n w_{t} - \hat{\rho} \bar{w}\sum_{t=1}^n y_{t-1}}{\bar{s}\sum_{t=1}^n w_{t} - n},\\
\hat {\mu} &= \bar{y} - \frac{\hat{\rho}}{n}\sum_{t=1}^{n} y_{t-1} - \hat{\beta} \bar{s},\\
\hat {\delta} &= \sqrt{\frac{\bar{s}}{(\bar{s} \bar{w} - 1)}},\;\; \hat {\gamma} = \frac{\hat{\delta}}{\bar{s}},\;\mbox{and}\;
\hat{\alpha} = (\hat{\gamma}^2 + \hat{\beta}^2)^{1/2},
\end{align*} 
where $\bar{s} = \displaystyle\frac{\sum_{t=1}^{n}s_t}{n}, \bar{w} = \frac{\sum_{t=1}^{n}w_t}{n} \text{ and } \bar{y} = \displaystyle\frac{\sum_{t=1}^{n}y_t}{n}.$

\noindent In next section, the efficacy of the estimation procedure is discussed on simulated data. Further, a real life application of the NIGAR(1) model is discussed.

\section{Simulation Study and Real World Data Application}
In this section, we discuss the parameter estimation of the proposed model on the simulated data. Further, the application of the proposed time-series model is shown on the Google equity price.\\
\noindent{\bf Simulation study:}
For simulation study, it is required to generate the iid inverse Gaussian distributed random numbers $G_i \sim IG(\mu_1, \lambda_1),\;i=1,2,\cdots,N$. We use the algorithm mentioned in \cite{Devroye1986}, which involves following steps:
\begin{itemize}
    \item []{\it Step 1}: Generate standard normal variate $N$ and set $Y = N^2$.
    \item[] {\it Step 2}: Set $X_1 = \mu_1 + \frac{\mu_{1}^2 Y}{ 2 \lambda_1} - \frac{\mu_1 }{2 \lambda_1}\sqrt{4\mu_1 \lambda_1 Y + \mu_1^2 Y^2}$.
    \item []{\it Step 3}: Generate  uniform random variate $U[0,1]$.
    \item[]{\it Step 4}: If $U<= \frac{\mu_1}{\mu_1 + X_1}$, then $G = X_1$; else $G = \frac{\mu_1^{2}}{X_1}.$ 
\end{itemize}
Note that the form of the pdf taken in Devroye (1986) is different than the form given in \eqref{IG_density}, so it is required to make the following substitutions for the parameters $\mu_1 = \delta/\gamma$ and $\lambda_1 = \delta^2.$ These substitutions using above steps will generate random numbers from IG$(\gamma, \delta)$ with pdf given in \eqref{IG_density}. In order to generate the i.i.d. NIG innovation terms, we use the normal variance-mean mixture form of the NIG with inverse Gaussian as mixing distribution and standard normal distribution representation given in \eqref{mean_variance_mixture}, which will give
$\varepsilon_{t} \sim NIG(\alpha, \beta, \mu, \delta)$.
A simulated series of NIGAR(1) model is plotted in Fig 1. It is intuitive from the plot that the series has roughly constant mean and constant variance over time with occasional large jumps. This kind of behaviour can't be explained by using standard Gaussian innovation term autoregressive model.


\begin{figure}[ht!]
\centering
\subfigure[]{
\includegraphics[width=7.8cm, height=7cm]{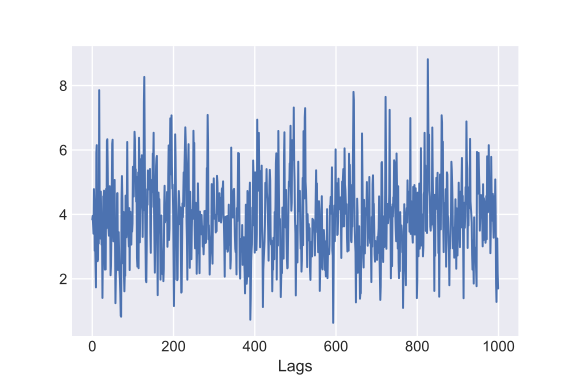}}
\quad
\subfigure[]{
\includegraphics[width=7.8cm, height= 7cm]{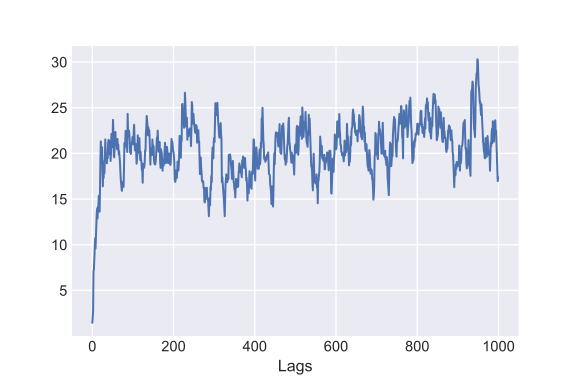}}
\caption{The figure shows the simulated NIGAR(1) series (size =1000) for parameters $\delta=2, \gamma =2, \mu =1, \beta =1$ with (a) $\rho = 0.5$ and (b) $\rho = 0.9$.}
\end{figure}

Next, we simulate an NIG random sample of size 10,000 with true parameter values as $ \alpha = 2.24, \beta = 1, \mu = 1$ and  $\delta = 2.$ Using the relationship $y_t = \rho y_{t-1} + \varepsilon_t$, the simulated data is used to generate the NIGAR(1) time series with $\rho = 0.5$ and $y_0 = \varepsilon_0$. To analyse the behavior of the EM algorithm, we use two different stopping criteria, first is based on the relative change of the log-likelihood (criterion 1) and the other is based on the relative change of the parameter values (criterion 2). Let $L^{(k)}$ denote the log-likelihood value after $k$ iterations, then we stop the iterations when $\abs{ \frac{L^{(k)} - L^{(k+1)}}{L^{(k)}}} < 10^{-5}.$ In other criterion, we stop the iterations when
$$
\max \left\{\abs{ \frac{\alpha^{(k+1)} - \alpha^{(k)}}{\alpha^{(k)}}}, \abs{ \frac{\beta^{(k+1)} - \beta^{(k)}}{\beta^{(k)}}}, \abs{ \frac{\mu^{(k+1)} - \mu^{(k)}}{\mu^{(k)}}}, \abs{ \frac{\gamma^{(k+1)} - \gamma^{(k)}}{\gamma^{(k)}}}, \abs{ \frac{\delta^{(k+1)} - \delta^{(k)}}{\delta^{(k)}}} \right\} < 10^{-5}.$$

\noindent \textbf{Case 1:} We use the parameter estimates derived from Prop. \ref{NIG_AR(1)} on simulated data. 
\noindent Using the first stopping criterion, the estimated values of parameters are $\hat{\alpha} = 2.073, \hat{\beta} = 0.880, \hat{\mu} = 1.077, \hat{\delta} = 1.953, \text{ and } \hat{\gamma} = 1.877.$

\noindent When second criterion is used the estimated values are $\hat{\alpha} = 1.980, \hat{\beta} = 0.824, \hat{\mu} = 1.125, \hat{\delta} = 1.897,  \text{ and }\hat{\gamma} = 1.800.$ In order to check the convergence of algorithm, we used different set of initial values which resulted in approximately same estimates. The results are summarized in following table:

\begin{center}
\begin{table}[ht!]
\centering
\begin{tabular}{||c||c||c|}
\hline 
 & Actual & Estimated \\ 
\hline 
Criterion 1 & $\alpha=2.24$, $\beta=1$, $\mu=1$, $\delta =2$ &$\hat{\alpha}=2.091$, $\hat{\beta}=0.892$, $\hat{\mu}=1.042$ , $\hat{\delta} = 1.962$\\ 
\hline 
Criterion 2&  $\alpha=2.24$, $\beta=1$, $\mu=1$, $\delta =2$ &$\hat{\alpha}=2.090$, $\hat{\beta}=0.892$, $\hat{\mu}=1.042$ , $\hat{\delta} = 1.962$\\ 
\hline 
\end{tabular} 
\caption{Actual and estimated parameters using different stopping criterion for case 1.}
\end{table}
\end{center}
\textbf{Case 2:} We use the parameter estimates of NIG proposed by Karlis \cite{Karlis2002} and estimate the $\rho$ using conditional least square estimate given in \eqref{cls_rho}.
The estimated values are tabulated below as:
\begin{center}
\begin{table}[ht!]
\centering
\begin{tabular}{||c||c||c|}
\hline 
 & Actual & Estimated \\ 
\hline 
Criterion 1 & $\alpha=2.24$, $\beta=1$, $\mu=1$, $\delta =2$ &$\hat{\alpha}=2.073$, $\hat{\beta}=0.880$, $\hat{\mu}=1.077$ , $\hat{\delta} = 1.953$\\ 
\hline 
Criterion 2&  $\alpha=2.24$, $\beta=1$, $\mu=1$, $\delta =2$ &$\hat{\alpha}=1.980$, $\hat{\beta}=0.824$, $\hat{\mu}=1.125$ , $\hat{\delta} = 1.897$\\ 
\hline 
\end{tabular} 
\caption{Actual and estimated parameters using different stopping criterion for Case 2.}
\end{table}
\end{center}

\noindent In Fig. 2, we display the box plots of the parameters estimated using EM algorithm on 100 simulated data sets with the help of Prop. \ref{NIG_AR(1)}. Observe that the parameters are symmetric and centered on $\hat{\alpha} = 2.22 , \hat{\beta} = 1.01 , \hat{\mu} = 0.99, \hat{\delta} = 1.99 \text{ and } \hat{\rho} = 0.501$. From the whiskers size we can observe that the distribution of the parameters may be heavy tailed.

\begin{figure}[ht!]
\centering
\includegraphics[width=12cm, height=8cm]{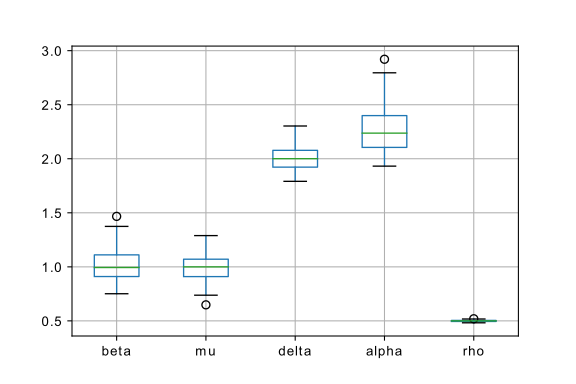}
\caption{Box plots of the estimates of parameters.}
\end{figure}



\noindent{\bf Real world data application:} The historical financial data of Google equity is collected from Yahoo finance. The whole data set covers the period from December 31, 2014 to April 30, 2021. Initially, the data contained 1594 data points with 6 features having Google stock's {\it open price}, {\it closing price}, {\it highest value}, {\it lowest value}, {\it adjusted close price} and {\it volume} of each working day end-of-the-day values. In order to apply the proposed NIGAR(1) model, we take the univariate time series $y_t$ to be the end-of-the-day closing prices. The Google equity closing price is demonstrated as time series data in Fig. 3. 
\begin{figure}[ht!]\centering
\includegraphics[width=12cm, height=8cm]{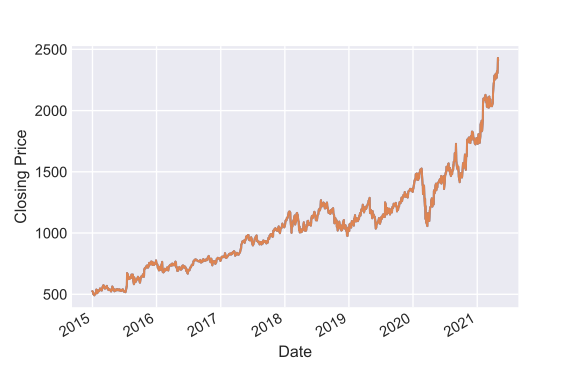}
\label{fig:clsprice}   
\caption{The closing price (in\$) of Google equity.}
\end{figure}

We assume that the innovation terms $\varepsilon_t$ of time series data $y_t$ follows NIG. We use ACF and PACF plot to determine the appropriate time series model components for closing price. Fig. 4 shows the ACF and PACF plot of time-series data $y_t$. 
\begin{figure}[ht!]
\centering
\subfigure[ACF Plot]{
\includegraphics[width=8cm, height=6cm]{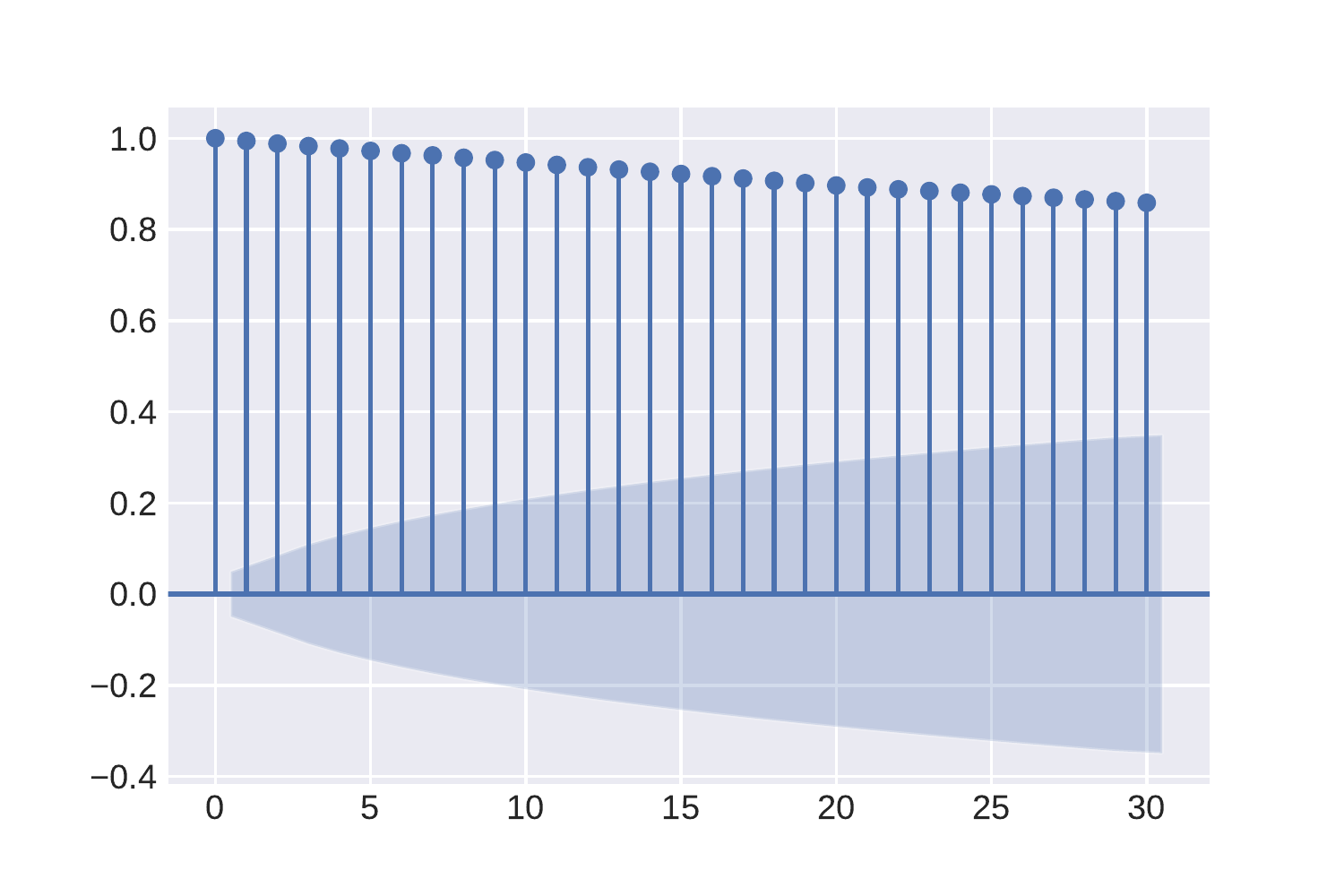}}
\subfigure[PACF Plot]{
\includegraphics[width=8cm, height= 6cm]{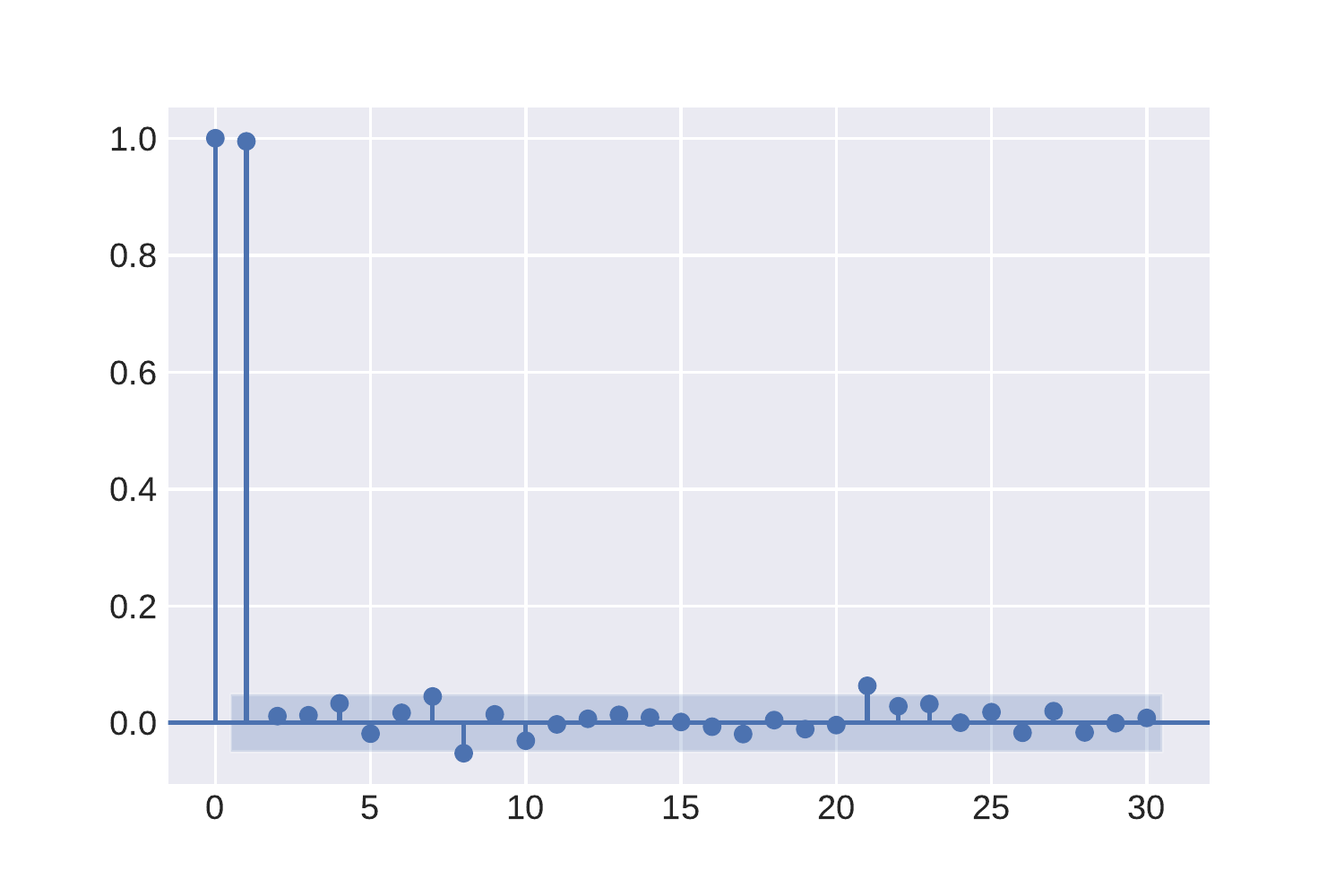}}
\caption{The ACF and PACF plot of closing price of Google equity.}
\end{figure}


We observe that in PACF plot there is a significant spike at lag 1, also we ignore the spike at lag 0 as it represents the correlation between the term itself which will always be 0. PACF plot indicates that the closing prices follow AR model with lag 1. In ACF plot, all the spikes are significant for lags upto 30 which implies that the closing price is highly correlated. Therefore, from ACF and PACF plot the assumed NIGAR(1) model is expected to be good fit for data.
First we estimate the $\rho$ parameter using the conditional least square method as mentioned in \eqref{cls_rho}
$$\hat{\rho} = \frac{\sum_{t=0}^{n}(y_{t} - \overline{y})(y_{t+1} - \overline{y})}{\sum_{t=0}^{n}(y_{t} - \overline{y})^2}.   $$
The estimated value is $\hat{\rho} = 0.9941$. Using the estimated value of $\rho$ and relation $ \varepsilon_t = y_t - \rho y_{t-1}$ we will get the innovation terms $\varepsilon_t.$ The distribution plot of innovation terms is shown in Fig. 5(a).\\

 
\noindent The Kolmogorov-Smirnov (KS) normality test and Jarque–Bera (JB) test are performed on $\varepsilon_t$ to test if the innovation terms are Gaussian. The $p$-value in both the tests was 0, which indicates that the $\varepsilon_t$ may not be from Gaussian distribution. Therefore, we fit the proposed NIGAR(1) model on Google closing price dataset. The ML estimates of parameters using EM algorithm are $\hat{\alpha} = 0.0202, \hat{\beta} = 0.0013, \hat{\mu} = 0.226, \hat{\delta} = 9.365, \text{ and } \hat{\gamma} = 0.0201$ with initial guesses as $\hat{\alpha}^{(0)} = 0.0141, \hat{\beta}^{(0)} = 0.01, \hat{\mu}^{(0)} = 0.01, \hat{\delta}^{(0)} = 0.01, \text{ and } \hat{\gamma}^{(0)} = 0.01.$ The relative change in the log-likelihood value with tolerance value 0.0001 is used as stopping criterion. It is worthwhile to mention that the estimated $\beta$ is close to 0 and the estimated $\mu$ can be interpreted as intercept term in the AR(1) model. Based on the Google equity prices data and estimated parameters, an equivalent model to \eqref{main_model} can be described as
$$
Y_t = \mu + \rho Y_{t-1} + \epsilon_t,
$$
where $\epsilon_t = \sqrt{G}Z$ with $G\sim$IG$(\gamma, \delta)$, with $\hat{\rho} = 0.9941, \hat{\mu} = 0.226, \hat{\delta} = 9.365, \text{ and } \hat{\gamma} = 0.0201.$\\
\begin{figure}[h!]
\centering
\subfigure[Distribution of error terms]{
\includegraphics[width=8cm, height=6cm]{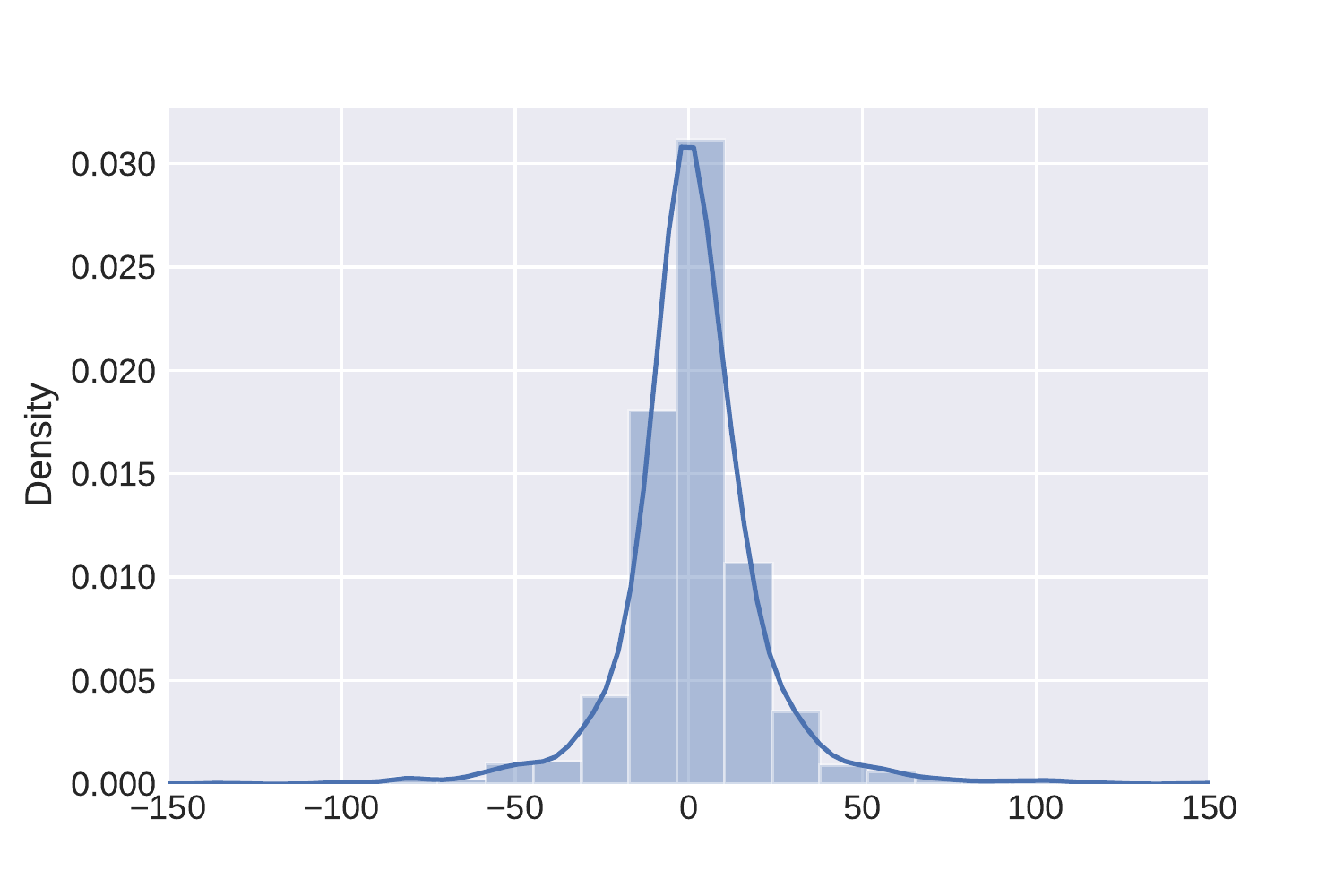}}
\quad
\subfigure[QQ Plot]{
\includegraphics[width=8cm, height= 6cm]{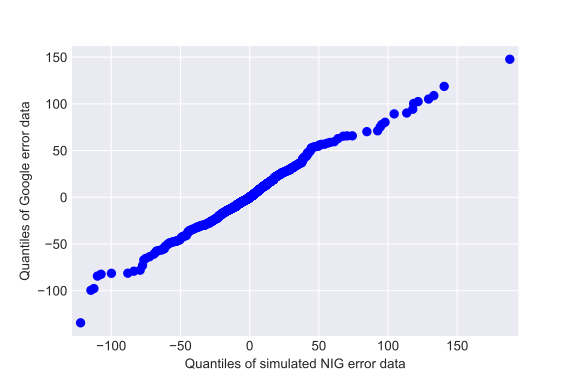}}
\caption{Plot of error terms distribution and QQ plot between simulated and actual values}
\end{figure}

To test the reliability of the estimated results, we used the simulated data with true parameter values as $\alpha = 0.02, \beta = 0, \mu = 0.23, \delta = 9.5, \text{ and } \gamma = 0.02.$ The QQ-plot (Fig. 5(b)) between the error terms $\varepsilon$ of data and the NIG simulated data indicates that the innovation terms are significantly NIG with few outliers. Also, we performed the 2-sample KS test which is used to check whether the two samples are from same distribution. The 2-sample KS test on the simulated data and the innovation terms resulted in the $p$-value$ = 0.3413.$ Therefore, we failed to reject the null hypothesis that the samples are from same distribution and conclude that both samples follow the same distribution i.e. NIG$(\alpha=0.02, \beta=0, \mu=0.23 , \delta=9.5)$.


\section{Conclusions}
The introduced AR(1) model with NIG innovations very well model the prices of the Google equity for the period from December 31, 2014 to April 30, 2021. Note that NIG distribution is semi heavy tailed and used to model equity returns in the literature. So, the proposed model could be useful in modeling time-series which manifests large observations where the classical AR(1) model with Gaussian-based innovations can not be used. We estimate the parameters using EM algorithm since direct maximization of likelihood function is difficult due to the modified Bessel function of the third kind term. The EM algorithm estimates the parameter efficiently and the results are verified on the simulated data.\\
\vtwo

\noindent{\bf Acknowledgements:}  Monika S. Dhull would like to thank the Ministry of Education (MoE), Government of India, for supporting her PhD research work.

\vone
\noindent



\begin{thebibliography}{9}
\bibitem{Barndorff1997a} Barndorff-Nielsen, O.E.: Exponentially decreasing distributions for the logarithm of particle size. Proc. Roy. Soc. London Ser. A., 353, 401-409 (1977).

\bibitem{Barndorff1997b} Barndorff-Nielsen, O.E.: Normal Inverse Gaussian Distributions and Stochastic Volatility Modelling. Scand. J. Stat., 24, 1-13 (1997).

\bibitem{Barndorff2013} Barndorff-Nielsen, O.E., Mikosch, T. and Resnick, S. I.: L\'evy Processes: Theory and Applications, Birkhauser Basel (2013).

\bibitem{Christmas2011} Christmas, J. and Everson, R.: Robust autoregression: Student-t innovations using variational Bayes. IEEE Trans. Signal Process., 59, 48-57 (2011).

\bibitem{Cont2004} Cont, R. and Tankov, P.: Financial Modeling With Jump Processes. Chapman \& Hall, CRC Press, London, UK (2004).

\bibitem{Dempster1977} Dempster, A.P., Laird, N.M. and Rubin, D.B.: Maximum likelihood from incomplete data via EM Algorithm. J. R. Stat. Soc. Ser. B 39, 1-38 (1977).

\bibitem{Devroye1986} Devroye, L.: Non-Uniform Random Variate Generation( 1st ed.), Springer, New York (1986).

\bibitem{Do2008} Do, C. B. and Batzoglou, S.: {What is the expectation maximization algorithm?}. {Nature Biotechnology.} {26(8)}, 897-899 (2008).

\bibitem{Grunwald1996} Grunwald,   G., Hyndman,   R. and  Tedesco,   L.:  A  unified  view  of  linear  AR(1) models. D.P. Monash University, Clayton (1996).

\bibitem{Heyde2005} Heyde, C. C. and Leonenko, N. N.: Student processes. Adv. Appl. Probab., 37 342-365 (2005). 

\bibitem{Jorgensen1982} J\o rgensen, B.: Statistical properties of the generalized inverse Gaussian distribution. Lecture Notes in Statistics, vol 9. Springer-Verlag, New York (1982).

\bibitem{Kalemanova2007} Kalemanova, A., Schmid, B. and Werner, R.: The Normal inverse Gaussian distribution for synthetic CDO pricing. J. Deriv., 14, 80-93 (2007).

\bibitem{Karlis2002} Karlis, D.: {An EM type algorithm for maximum likelihood estimation of the normal-inverse Gaussian distribution}. {Statist. Probab. Lett.,} {57}, 43-52 (2002).

\bibitem{Kim2008} Kim, J., Stoffer, D. S.: Fitting stochastic volatility models in the presence of irregular sampling via particle methods and the EM algorithm. J. Time Series Anal., 29, 811–833 (2008).

\bibitem{Liu1995} Liu, C. H., Rubin, D. B.: ML estimation of $t$-distribution using EM and its extensions, ECM and ECME. Statist. Sinica., 5, 19–39 (1995).

\bibitem{Liu2019} Liu J., Kumar, S. and D. P. Palomar, D. P.: {Parameter Estimation of Heavy-Tailed AR Model With Missing Data Via Stochastic EM}. IEEE Trans. Signal Process., 67 vol.8, 2159-2172 (2019).

\bibitem{McLachlan2007} McLachlan, G. J., Krishnan, T.: The EM Algorithm and Extensions (2nd ed.), John Wiley \& Sons, New Jersey (2007).

\bibitem{Meitz2021} Meitz, M., Preve, D. and Saikkonen, P.: A mixture autoregressive model based on Student’s t–distribution. Commun. Stat. Theory Methods (2021).

\bibitem{Meng1993} Meng, X. L., Rubin, D. B.: Maximum likelihood estimation via the ECM algorithm: A general framework. Biometrika, 80, 267–278 (1993).

\bibitem{Metaxoglou2007} Metaxoglou, K., Smith, A.: Maximum Likelihood Estimation of VARMA Models Using a State-Space EM Algorithm. J. Time Series Anal., 28, 666–685 (2007).

\bibitem{Nduka2018} Nduka, U. C.:  EM-based algorithms for autoregressive models with $t$-distributed innovations. Commun. Statist. Simul. Comput., 47, 206-228 (2018).

\bibitem{Nitithumbundit2015} Nitithumbundit, T. and Chan, J. S. K.: An ECM algorithm for Skewed Multivariate Variance Gamma Distribution in Normal Mean-Variance Representation (2015). arXiv  preprint arXiv:1504.01239.

\bibitem{Hagan2016} O'Hagan, A., Murphy, T. B., Gormley, E. C. et. al.: {Clustering with the multivariate normal inverse Gaussian distribution}. { Comput. Statist. Data Anal.}. {93}, 18-30 (2016).

\bibitem{Oigard2005} \O igard, T. A., Hanssen, A., Hansen, R. E. and Godtliebsen, F.: EM-estimation and modeling of heavy-tailed processes with the multivariate normal inverse Gaussian distribution. Signal Processing, 85, 1655-1673 (2005).

\bibitem{Omeya2018} Omeya, E. Gulcka, S. V. and  Vesilo, R.: Semi-heavy tails, Lith. Math. J., 58, 480-499 (2018).

\bibitem{Protassov2004} Protassov, R.S.: EM-based maximum likelihood parameter estimation for multivariate generalized hyperbolic distributions with fixed $\lambda$. Stat. Comput., 14, 67–77 (2004).

\bibitem{Rachev2003} Rachev, S.T.: Handbook of Heavy Tailed Distributions in Finance: Hand-books in Finance. Amsterdam, The Netherlands: Elsevier (2003).

\bibitem{Sim1990} Sim, C. H.:  First-order autoregressive models for gamma and exponential processes.  J. Appl. Prob., 27, 325-332 (1990).

\bibitem{Tarami2003} Tarami, B. and Pourahmadi, M.: Multi-variate t autoregressions: Innovations, prediction variances and exact likelihood equations. J. Time Ser. Anal., 24, 739-754 (2003).

\bibitem{Tiku2000} Tiku, M. L., Wong, W.-K.,  Vaughan, D. C.  and Bian, G.:  Time series models in non-normal situations: Symmetric innovations. J. Time Ser. Anal., 21, 571-596 (2000).

\bibitem{Tsay2005} Tsay, R.S.: Analysis of Financial Time Series (2nd ed.), Hoboken, NJ, USA, Wiley (2005).

\end{thebibliography}
\end{document}